\documentclass[letterpaper,twocolumn,10pt]{article}
\usepackage{usenix2019_v3}

\usepackage{amsmath,amssymb,amsfonts,amsthm}
\usepackage{graphicx}
\usepackage{booktabs}
\usepackage{tabularx}
\usepackage{array}
\usepackage{float}
\usepackage{listings}
\usepackage{algorithm}
\usepackage{algpseudocode}
\usepackage{subcaption}
\usepackage{xspace}
\usepackage{tikz}
\usetikzlibrary{arrows.meta, positioning, calc, shapes.geometric, fit, shadows, matrix}

\microtypesetup{spacing=false}

\AtBeginDocument{\DeclareMathAlphabet{\mathcal}{OMS}{cmsy}{m}{n}}

\newtheorem{proposition}{Proposition}
\newtheorem{definition}{Definition}

\newcolumntype{L}[1]{>{\raggedright\arraybackslash}p{#1}}
\newcolumntype{C}[1]{>{\centering\arraybackslash}p{#1}}
\newcolumntype{R}[1]{>{\raggedleft\arraybackslash}p{#1}}

\definecolor{navy}{RGB}{20,50,110}
\definecolor{crimson}{RGB}{180,30,30}
\definecolor{slate}{RGB}{112,128,144}
\definecolor{emerald}{RGB}{34,139,34}
\definecolor{darkpurple}{RGB}{102,51,153}
\definecolor{sublayer}{RGB}{235,240,250}
\definecolor{kernellayer}{RGB}{250,235,235}

\hypersetup{
  pdftitle={Stored Is Not Supported: Typed Provenance and Assertion Guardrails for Persistent AI Agents},
  pdfauthor={Jun He; Deying Yu},
  pdfsubject={Typed Provenance, Assertion Guardrails, Persistent AI Agents},
  pdfkeywords={persistent AI agents, autobiographical claims, epistemic guardrails, typed provenance, assertion mediation, memory poisoning, AI safety}
}

\begin{document}

\title{\bf Stored Is Not Supported:\\
Typed Provenance and Assertion Guardrails for Persistent AI Agents}

\author{
  {\rm Jun He}\\
  OpenKedge.io
  \and
  {\rm Deying Yu}\\
  OpenKedge.io
}

\maketitle

\begin{abstract}
Persistent AI agents construct autobiographical state through reflection, retrieval, and
consolidation. Persistence changes availability, not epistemic standing: stored or retrieved
material is not thereby supported. Untrusted inputs, prompt injections, and model inferences can
therefore enter persistent state and later be presented as agent history or user commitments. We
specify typed provenance and assertion guardrails for \emph{autobiographical assertion boundedness},
a system-relative release property requiring governed statements about the agent, user, or named
relationships to satisfy accepted-evidence, temporal-validity, and disclosure policies. A typed provenance graph separates origin, dependency
lineage, epistemic role, validity, and disclosure scope. A resolver evaluates authorized state
projections and returns one evidential status, orthogonal conflict, staleness, and withholding flags,
and a protected decision witness. A generate--verify--revise mediator then checks candidate semantic units before release
and renders policy-authorized status responses. Under explicit assumptions about extraction,
predicate correctness, resolution soundness, view declassification, and channel mediation, we prove
a conditional assertion-boundedness contract. In an executable suite of 24 hand-authored
conformance cases, typed mediation passed none of 19 unsafe opportunities unqualified while
preserving all five supported controls. The flat/prior and source-tag comparison rules released
19/19 and 18/19 unsafe candidates, respectively. These results validate the encoded resolver
and mediator obligations; they do not constitute an end-to-end evaluation of language models or
retrieval systems.
\end{abstract}

\section{Introduction}
\label{sec:introduction}

Autonomous AI agents increasingly maintain long-term state across extended interactions,
employing hierarchical memory architectures that combine periodic reflection, semantic
retrieval, and background consolidation~\cite{park2023generative,packer2023memgpt,xu2025amem,chhikara2025mem0}.
Persistence expands what an agent can retrieve without establishing what it may assert. Benchmarks
such as LoCoMo evaluate factual retrieval over extended conversational
horizons~\cite{maharana2024locomo}; retrieval alone cannot confer epistemic entitlement. Storing or
retrieving an unverified string or model inference does not authorize its presentation as an
authentic part of the agent's autobiographical history.

Consider an agent that ingests an untrusted document containing an embedded indirect instruction
asserting that the user prefers public data disclosure. During background memory consolidation,
a language model may distill this unverified instruction into a persistent user profile. In a
subsequent session, the agent retrieves the isolated summary and asserts with first-person
authority: ``You asked me to share this information publicly.'' Through recurrent retrieval and
re-summarization loops, such ungrounded claims can self-corroborate, while stronger successor
models may embellish them with plausible parametric details. Prompt-injection evaluations show
that untrusted inputs can steer agent reasoning~\cite{debenedetti2024agentdojo}; in a persistent
architecture, a resulting unsupported claim can survive the originating context and influence
later agent-specific responses.

Common mitigations---such as tracking source URLs, scoring model confidence, or prompting models
for self-reflection---do not jointly enforce the policy dimensions considered here. A
cryptographically authenticated source may supply inaccurate information; multiple retrieved
records may trace back to a single upstream root, manufacturing an illusion of consensus;
historically valid preferences can expire; and legitimately supported facts may be restricted
from specific recipients. Crucially, procedural entitlement requires decoupling orthogonal
operational dimensions: state admission (reachability from an authoritative head), evidential
support (conformance to explicit dependency and derivation policies), agent stance (recorded
belief versus attributed report), temporal validity (active versus expired), and disclosure
authorization (recipient and purpose access control). None of these checks decides metaphysical
truth; rather, they enforce auditable procedural bounds over what an agent is permitted to claim.

We address the central question:
\begin{center}
\emph{Can a persistent agent distinguish what it stores or its underlying model generally knows
from what this specific agent has evidence, authority, and valid scope to claim?}
\end{center}
We formalize \emph{autobiographical assertion boundedness}, a system-level release property
governing outward statements concerning the agent, its user, or named relationships. Rather than
attempting the intractable task of pruning a model's latent parametric knowledge, our
architecture treats storage and support as separate judgments by regulating transitions between
explicit epistemic roles. Model priors, raw retrieved inputs, admitted evidence, derived
interpretations, and recorded beliefs remain distinct persisted objects. Outward autobiographical
assertions are mediated through an explicit release gate: unqualified content assertions must bind
to admitted, policy-supported, current, and disclosure-authorized claims at an exact state head,
whereas resolution-status responses must safely render internal resolver outcomes without leaking
protected metadata.

Separating storage from support requires more than a well-formed memory store: an opaque model can
hallucinate biographical claims even when the store is flawless. We therefore establish three
distinct enforcement surfaces:
\begin{itemize}
  \item \textbf{State admission:} Governs which claim-bearing objects become authoritative in
  the agent's state store;
  \item \textbf{Authorized projection:} Filters accepted state into task- and recipient-specific
  views prior to model context assembly;
  \item \textbf{Assertion mediation:} Intercepts candidate responses at governed output channels
  to verify content support and disclosure safety before release.
\end{itemize}
Together, these surfaces prevent persistence from silently converting availability into assertion
authority.

\paragraph{Accepted-state boundary.}
Transactional state mechanisms govern which versioned objects become reachable at an
authoritative head~\cite{he2026continuitykernel,li2026memtx}. This paper takes that accepted head
as an authenticated input and addresses a different question: which reachable claims support a
particular outward statement for a particular recipient and time. Atomic storage therefore remains
an assumption at the lower boundary, not a contribution or evaluation target of this work.
Reachability supplies candidates for resolution; it does not establish support by itself.

\paragraph{Contributions.}
This paper makes four contributions:
\begin{enumerate}
  \item \textbf{Typed autobiographical claim model.} We formulate a typed provenance graph over
  evidence and claim vertices, with policy-relative root recognition and role promotion that
  separate origin, dependency lineage, epistemic role, validity time, and disclosure scope
  (Section~\ref{sec:epistemic-foundations}).
  \item \textbf{Procedural resolution and forward correction.} We design a typed resolver that
  evaluates authorized projections to return a single evidential status alongside orthogonal
  conflict, staleness, and withholding flags, generating a minimized, auditable decision witness
  and supporting dependency-aware invalidation (Section~\ref{sec:epistemic-outcomes}).
  \item \textbf{Assertion mediation pipeline.} We formulate a generate--verify--revise release
  gate and prove a conditional contract result for content assertions and recipient-safe status
  renderings under emission-head verification (Section~\ref{sec:assertion-mediation}).
  \item \textbf{Conformance validation and evaluation design.} We provide an executable
  conformance suite covering 24 structured cases across six threat tracks, together with metrics
  and a separate protocol for future end-to-end evaluation
  (Section~\ref{sec:epistemic-threats-evaluation}).
\end{enumerate}

\paragraph{Scope and boundaries.}
Our guarantees apply to declared autobiographical claims on mediated output channels over an
intact accepted head. The executable artifact tests structured resolver and mediator behavior; it
does not implement natural-language semantic extraction, generalized witness access control,
generalized view-equivalence checks, full provenance-graph traversal, canonical assertion-artifact
hashing, or target-origin admission checks. The mechanism also does
not implement lower-level concurrency or crash recovery, certify external truth, audit unmediated
out-of-band channels, or make claims regarding machine consciousness or subjective identity.

\section{Epistemic State Model and Threat Boundary}
\label{sec:epistemic-foundations}

We model the epistemic state of an autonomous agent over an abstract, authenticated
state-transition interface. The results below are conditional on four environment assumptions:
\begin{itemize}
  \item \textbf{(E1) Authoritative state ordering:} A monotonically versioned branch head $h$
  identifies a finite reachable set of accepted, immutable epistemic objects, separating
  them from uncommitted, speculative, or rejected candidates.
  \item \textbf{(E2) Versioned dependency closure:} Every evidence, claim, and derivation
  identifier within the accepted state resolves deterministically to an immutable content version
  or an explicit, authenticated unavailable marker.
  \item \textbf{(E3) Immutable policy binding:} Epistemic admission, projection, and release
  policies $\kappa$ are versioned, authenticated, and bound to execution contexts.
  \item \textbf{(E4) Closed channel mediation:} Every governed outbound communication channel
  routes candidate agent emissions through the assertion mediator without unmonitored bypass paths.
\end{itemize}
Append-only logs, Merkle DAGs, and transactional stores can realize E1--E2
(e.g.,~\cite{he2026continuitykernel,li2026memtx}). Our theoretical
formulation treats the accepted head $h$ as an abstract authenticated input. The model therefore
focuses on epistemic classification, dependency topology, and warranted assertibility.

\subsection{Covered Autobiographical Language}

Let $\mathcal L_{bio} \subset \mathcal L$ denote the covered formal sublanguage of structured
autobiographical propositions. A proposition $\phi \in \mathcal L_{bio}$ is a typed tuple
specifying a subject, predicate, argument terms, temporal interval, and epistemic modality,
where:
\begin{enumerate}
  \item The subject designates the agent itself, its authenticated principal (user), or a named
  relational entity within their shared operational environment; and
  \item The predicate characterizes an experiential interaction, stated preference, granted
  permission, relational commitment, task capability, or historical event.
\end{enumerate}
Under this formulation, autobiographical standing is semantic and system-relative rather than
purely grammatical. Assertions such as ``the user requested plain-text formatting'' or ``the agent
holds active consent for calendar inspection'' belong to $\mathcal L_{bio}$ irrespective of surface
syntactic framing or pronoun choice.

Deployments parameterize $\mathcal L_{bio}$ via an explicit domain signature
$\Sigma=\langle\mathcal P_{\mathrm{pred}},\mathcal A,\mathcal M\rangle$ defining permitted predicates
$\mathcal P_{\mathrm{pred}}$, argument sorts $\mathcal A$, and modal operators $\mathcal M$.
General world knowledge (e.g., mathematical truths or physical laws) remains outside $\mathcal L_{bio}$
unless explicitly framed as an agent-specific episodic experience or commitment (e.g., ``the agent
retrieved this document during yesterday's session'').

\subsection{Model Parametric Capacity vs. Accepted Agent Claims}

A persistent agent couples an underlying generative model substrate $m$ with an accepted state
head $h$. We distinguish two information spaces with different evidential status:
\begin{enumerate}
  \item \textbf{Accepted Agent Claim Space:} For any state head $h$, let
  \begin{equation}
  \mathcal C^{\mathrm{agent}}(h) = \operatorname{ClaimsReachable}(h)
  \label{eq:agent-claim-set}
  \end{equation}
  denote the finite, enumerable set of typed claim objects accepted into the agent's
  authoritative state at $h$.
  \item \textbf{Latent Model Parametric Space:} Let
  \begin{equation}
  \begin{aligned}
  \mathcal I^{\mathrm{model}}_m = \{x \mid{} &x\text{ is latent, contextual, or} \\
  &\text{retrieved information accessible to } m\}
  \end{aligned}
  \label{eq:model-information}
  \end{equation}
  represent the opaque space of statistical correlations, pre-trained weights, and prompt-context
  associations accessible to model $m$. The release mechanism does not assume that this space is
  available as an enumerable, authenticated set of claims.
\end{enumerate}

Let $\Phi^{\mathrm{agent}}(h)=\{\phi(v)\mid v\in\mathcal C^{\mathrm{agent}}(h)\}$ be the propositions
carried by accepted claim objects. Rather than attempting to prune or retrain
$\mathcal I^{\mathrm{model}}_m$, the release contract uses the following non-entitlement principle:
\begin{equation}
\operatorname{CanGenerate}_m(\phi) \not\Rightarrow \phi \in \Phi^{\mathrm{agent}}(h).
\label{eq:model-nonentitlement}
\end{equation}
The model's ability to generate a proposition does not give that proposition standing in the
agent's historical state. Consequently, upgrading or replacing the model
substrate ($m_1 \to m_2$) under a fixed state head $h$ may alter linguistic fluency, general
reasoning, or heuristic planning, but leaves the accepted autobiographical claim set $\mathcal C^{\mathrm{agent}}(h)$
invariant.

\subsection{Decoupling Truth, Provenance, and Warranted Assertibility}

Semantic truth and procedural support are distinct predicates. Let
$\operatorname{True}(\phi)$ denote an external truth predicate that the mechanism does not decide.
For a claim vertex $v$ with $\phi(v)\in\mathcal L_{bio}$, evaluated at head $h$ under policy
$\kappa$, the model assumes none of the following implications:
\begin{align}
\operatorname{Admitted}(v,h) &\not\Rightarrow \operatorname{True}(\phi(v)), \\
\operatorname{Authenticated}(v) &\not\Rightarrow \operatorname{True}(\phi(v)), \\
\operatorname{Supported}(v,h;\kappa) &\not\Rightarrow \operatorname{True}(\phi(v)), \\
\neg\operatorname{Supported}(v,h;\kappa) &\not\Rightarrow \neg\operatorname{True}(\phi(v)).
\label{eq:support-not-truth}
\end{align}
Here, $\operatorname{Admitted}(v,h)$ denotes reachability within the state store;
$\operatorname{Authenticated}(v)$ confirms cryptographic custody of the recorded origin; and
$\operatorname{Supported}(v,h;\kappa)$ denotes satisfaction of the governing evidence and
derivation policy.

Table~\ref{tab:epistemic-dimensions} summarizes these decoupled dimensions. Cryptographic signatures
guarantee message integrity, not informant veracity; retrieval frequency reflects index topology
rather than independent epistemic corroboration; model generation confidence measures probability
calibration rather than evidential proof; and admitting a statement into memory records the event of
its receipt rather than certifying the objective reality of its content.

\begin{table*}[t]
\centering
\caption{Decoupled epistemic dimensions, formal properties, and operational limits.}
\label{tab:epistemic-dimensions}
\small
\begin{tabularx}{\textwidth}{L{2.8cm} L{5.4cm} X}
\toprule
\textbf{Epistemic Dimension} & \textbf{Formal Question} & \textbf{Operational Limit} \\
\midrule
State Admission & Is the entity reachable at authoritative head $h$? & Does not establish objective truth or ongoing validity \\
Source Authentication & Did the declared principal cryptographically author the payload? & Does not establish informant honesty, accuracy, or source independence \\
Evidential Support & Does the entity's provenance DAG satisfy release policy $\kappa$? & Does not guarantee empirical infallibility or global completeness \\
Agent Stance & Is the entity typed as an observation, attributed report, or hypothesis? & Does not certify external correctness \\
Disclosure Authority & Does recipient $p$ and purpose $u$ satisfy scope predicate $\mathbb{S}$? & Does not imply public domain availability \\
Semantic Truth & Does the proposition accurately describe external physical reality? & Not decided by this mechanism; depends on external facts and source reliability \\
\bottomrule
\end{tabularx}
\end{table*}

\subsection{Epistemic State Spaces and Provenance DAGs}

We formalize the accepted epistemic state at head $h$ as a labeled, directed acyclic provenance
graph:
\begin{equation}
\mathcal G_h = (\mathcal V_h, \mathcal E_h, \lambda),
\label{eq:epistemic-graph}
\end{equation}
where each vertex $v \in \mathcal V_h$ represents an immutable epistemic entity (a raw evidence
artifact or a typed claim). A directed edge $(u,v)\in\mathcal E_h$ points from a dependency $u$ to
an entity $v$ that depends on it.

Let $\mathcal L_{bio}^{\bot}=\mathcal L_{bio}\cup\{\bot\}$, where $\bot$ marks an evidence object
that carries no autobiographical proposition, and let
$\mathcal P_{\mathrm{scope}}=2^{\mathcal P_{\mathrm{principal}}\times
\mathcal U_{\mathrm{purpose}}\times\mathcal K_{\mathrm{task}}}$ be the authorization-scope domain.
The valuation function
$\lambda: \mathcal V_h \to \mathcal T \times \mathcal L_{bio}^{\bot} \times \mathcal I_{\text{time}} \times \mathcal P_{\text{scope}}$
assigns each vertex its structural epistemic properties:
\begin{equation}
\lambda(v) = \bigl(\tau(v), \phi(v), \mathbb T(v), \mathbb S(v)\bigr),
\end{equation}
where:
\begin{itemize}
  \item $\tau(v) \in \mathcal T$ designates the explicit epistemic role;
  \item $\phi(v) \in \mathcal L_{bio}^{\bot}$ is the structured proposition carried by $v$, or
  $\bot$ for non-propositional evidence;
  \item $\mathbb T(v)=[t_{\text{start}},t_{\text{end}}]\in\mathcal I_{\text{time}}$ is its
  validity interval over the extended time domain $\overline{\mathbb R}$; and
  \item $\mathbb S(v)\in\mathcal P_{\mathrm{scope}}$ defines the authorization scope across
  principals, purposes, and tasks.
\end{itemize}

Let $\mathcal M_h$ denote the authenticated unavailable markers returned by E2. These markers are
dependency references rather than epistemic vertices: they carry no proposition or role, and a
transition record binds each marker to the dependency version it replaces.

Root vertices with indegree zero in $\mathcal G_h$ represent foundational source objects (e.g.,
sensor capture logs, signed user messages, imported external documents). Vertices with non-zero
indegree bind an immutable derivation descriptor specifying the transition rule, policy version,
and evaluator configuration.

\subsection{Source Independence via Provenance Graph Disjointness}

A primary failure mode in persistent agent architectures is \emph{false corroboration}: multiple
retrieved summaries, search snippets, or paraphrased records derived from a single upstream source
are treated as independent consensus. In our formalization, multi-source corroboration policies
must evaluate structural disjointness over the provenance graph rather than counting raw records.

Let $\operatorname{Anc}_h(v)=\{u\in\mathcal V_h\mid u\leadsto v\}$ denote the reflexive,
transitive ancestors of $v$, and let
$\mathcal R(h)=\{u\in\mathcal V_h\mid\operatorname{indeg}(u)=0\}$ be the source roots. The
policy-recognized roots of $v$ are:
\begin{equation}
\operatorname{Roots}_\kappa(v,h)=
\{r\in\operatorname{Anc}_h(v)\cap\mathcal R(h)\mid\operatorname{RootOK}_\kappa(r,h)\}.
\label{eq:recognized-roots}
\end{equation}

\begin{definition}[Source Independence]
Two evidence or claim entities $u, v \in \mathcal V_h$ are policy-independent under $\kappa$,
denoted $\operatorname{Indep}_\kappa(u,v;h)$, if and only if their terminal source root sets are
nonempty and mutually disjoint. Writing $R_u=\operatorname{Roots}_\kappa(u,h)$ and
$R_v=\operatorname{Roots}_\kappa(v,h)$:
\begin{equation}
\operatorname{Indep}_\kappa(u,v;h) \iff
R_u\ne\varnothing\land R_v\ne\varnothing\land R_u\cap R_v=\varnothing.
\label{eq:source-independence}
\end{equation}
\end{definition}
This graph-theoretic definition is conservative: multiple URLs, distinct semantic embeddings, or
syntactic paraphrases that converge onto a shared upstream document share source roots and cannot
satisfy independent corroboration requirements.

\subsection{Deductive Admission and Role Promotion Calculus}

Moving an entity across epistemic roles is a governed state-transition judgment rather than an
implicit byproduct of language generation. We use the disjoint role partition
$\mathcal T=\mathcal X\mathbin{\dot\cup}\mathcal D\mathbin{\dot\cup}\mathcal U$, where:
\begin{itemize}
  \item $\mathcal X$ contains $\textsf{ObservedEvent}$ and $\textsf{AttributedReport}$, the
  grounded historical roles anchored in direct sensors or authenticated channels;
  \item $\mathcal D$ contains $\textsf{DeterministicDerivation}$,
  $\textsf{MemoryInterpretation}$, and $\textsf{RecordedBelief}$, which retain their dependencies;
  \item $\mathcal U$ contains $\textsf{ModelPrior}$, $\textsf{UnverifiedRetrieval}$,
  $\textsf{ProbabilisticInterpretation}$, and $\textsf{Hypothesis}$, the synthetic or unverified roles.
\end{itemize}

The judgment
$\mathcal G_h\vdash_\kappa v_s\xrightarrow[D]{\mathrm{prom}}v_t:r_t$ means that policy $\kappa$
permits a fresh target $v_t$ with role $r_t$ and exact dependency set
$D\subseteq\mathcal V_h\cup\mathcal M_h$. Admission at the successor head adds $v_t$, adds an edge
$(u,v_t)$ for each $u\in D\cap\mathcal V_h$, and retains any marker in $D\cap\mathcal M_h$ in the
transition record. Two illustrative rules are:
\begin{equation}
\begin{gathered}
\frac{\begin{gathered}v_s\in D\subseteq\mathcal V_h\cup\mathcal M_h,\quad
\tau(v_s)=\textsf{AttributedReport}\\
\operatorname{Fresh}_h(v_t),\quad\operatorname{EvalOK}_\kappa(D)
\end{gathered}}
{\mathcal G_h\vdash_\kappa v_s\xrightarrow[D]{\mathrm{prom}}
v_t:\textsf{MemoryInterpretation}}
\\[2mm]
\frac{\begin{gathered}v_s\in D\subseteq\mathcal V_h\cup\mathcal M_h,\quad
\tau(v_s)=\textsf{ModelPrior}\\
\operatorname{Fresh}_h(v_t)
\end{gathered}}
{\mathcal G_h\vdash_\kappa v_s\xrightarrow[D]{\mathrm{prom}}v_t:\textsf{Hypothesis}}
\end{gathered}
\label{eq:promotion-relation}
\end{equation}

Promotion never mutates or overwrites predecessor vertices. The strict non-promotion condition is:
\begin{equation}
\begin{aligned}
&\forall r_t\in\mathcal X,\ D\subseteq\mathcal V_h\cup\mathcal M_h:\\
&\quad\bigl(\exists u\in D\cap\mathcal V_h:\tau(u)\in\mathcal U\bigr)\\
&\qquad\implies\forall v_s\in D\cap\mathcal V_h:\
\mathcal G_h\nvdash_\kappa v_s\xrightarrow[D]{\mathrm{prom}}v_t:r_t.
\end{aligned}
\label{eq:no-silent-promotion}
\end{equation}
Furthermore, admitting a new root vertex $v$ directly into a grounded historical role $\mathcal X$
imposes the following root-admission obligations:
\begin{equation}
\begin{aligned}
&\forall v\in\mathcal V_h:\quad
\operatorname{indeg}(v)=0 \land \tau(v)=\textsf{ObservedEvent}\\
&\qquad\implies \operatorname{CaptureOriginOK}_\kappa(v,h), \\
&\forall v\in\mathcal V_h:\quad
\operatorname{indeg}(v)=0 \land \tau(v)=\textsf{AttributedReport}\\
&\qquad\implies \operatorname{ReportOriginOK}_\kappa(v,h).
\end{aligned}
\label{eq:target-origin-obligations}
\end{equation}
$\operatorname{CaptureOriginOK}$ checks the admitted capture path, device identity, content
commitment, and timestamp. $\operatorname{ReportOriginOK}$ checks an authenticated user- or
third-party-authored message, its channel custody, bytes, and reception context. Such a message can
ground an attributed report, never a direct observation.

\begin{proposition}[Epistemic Provenance Traceability]
\label{prop:promotion-traceability}
Assume E1--E3; a finite, acyclic $\mathcal G_h$; an admission gate that represents every non-root
vertex transition by the judgment above with its exact dependency set; and root admission that enforces
$\operatorname{RootOK}_\kappa$, including Equation~\ref{eq:target-origin-obligations} for roles in
$\mathcal X$. Every reachable claim $v\in\mathcal C^{\mathrm{agent}}(h)$ produced by admitted
transitions has a finite dependency trace obtained by recursively following its recorded references.
The trace terminates only at policy-recognized source artifacts or explicit unavailable markers,
and every admitted version and transition in it is inspectable at $h$.
\end{proposition}

\begin{proof}[Proof sketch]
Order $\mathcal G_h$ topologically and prove the stronger invariant for every admitted vertex. A
root satisfies the root-admission assumption; a grounded root
also satisfies Equation~\ref{eq:target-origin-obligations}. For a non-root $v$, the recorded
transition identifies a dependency set $D$, and E2 resolves each dependency to its exact version or
an unavailable marker. For an admitted version $u\in D\cap\mathcal V_h$, the induction hypothesis supplies its
finite trace, to which the recorded edge $(u,v)$ attaches $v$. An unavailable marker instead forms
a terminal trace entry bound to the affected dependency reference. Finiteness and acyclicity follow
from the graph assumption. Equation~\ref{eq:no-silent-promotion} separately prevents a source in
$\mathcal U$ from being retyped as a member of $\mathcal X$. Appendix~\ref{sec:proof-traceability}
provides the proof details.
\end{proof}

\subsection{Adversarial Threat Model}

We consider an active adversary capable of injecting malicious natural-language instructions,
corrupted biographical facts, or poisoned context into external web documents, retrieved snippets,
user dialogues, or third-party tool outputs~\cite{debenedetti2024agentdojo,sleeper2026poisoning}.
The adversary's objective is to induce the agent to emit false autobiographical statements, adopt
fabricated user preferences, or violate confidentiality scopes.

Specifically, the adversary may:
\begin{enumerate}
  \item \textbf{Sybil Corroboration:} Fabricate redundant records across distinct domains that
  originate from a single poisoned source to simulate synthetic consensus;
  \item \textbf{Consolidation Injection:} Exploit background summarization or memory reflection to
  silently promote unverified prompt inputs into permanent profile beliefs;
  \item \textbf{Temporal Exploitation:} Exploit expired permissions or stale historical preferences
  to authorize unauthorized present actions; or
  \item \textbf{Declassification Probing:} Query the agent with crafted prompts designed to leak
  restricted user data or the existence of withheld records.
\end{enumerate}

\paragraph{Trust Assumptions.}
We assume the integrity of the accepted state head $h$ (governed by transactional consensus), the
cryptographic key infrastructure, and the isolated execution environment of the assertion mediator.
Adversarial attacks that compromise private signing keys, subvert local operating system kernel
isolation, or leak information through unmediated side channels remain outside the primary security
boundary.

\section{Epistemic Resolution, Decision Witnesses, and Correction}
\label{sec:epistemic-outcomes}

Epistemic resolution formalizes the evaluation of a structured autobiographical query against an
authenticated state head, an authorized context projection, and an explicit release policy.
For a requesting principal $p$, purpose $u$, and operational scope $s$, let $\mathcal G_h = (\mathcal V_h, \mathcal E_h, \lambda)$
be the accepted state graph at head $h$. We define:
\begin{equation}
\begin{aligned}
\mathcal A_h &= \{v \in \mathcal V_h \mid \operatorname{Reachable}(v, h)\}, \\
\mathcal B_h(p, u, s) &= \{v \in \mathcal A_h \mid \operatorname{Authorized}(v, p, u, s)\},
\end{aligned}
\label{eq:authorized-claim-set}
\end{equation}
where $\mathcal A_h$ is the accepted epistemic state relative to $h$ and $\mathcal B_h(p,u,s)$ is
its authorized projection. Neither set is assumed to be deductively closed or globally
consistent; both may accommodate historical hypotheses, competing reports, unsuperseded conflicts,
and access-sealed references.

\subsection{Epistemic Resolution Semantics}

Let $Q \in \mathcal L_{bio}$ be a structured autobiographical query binding subject, predicate,
arguments, temporal bounds, and modality. Epistemic resolution under release policy $\kappa$ is
modeled as an evaluation judgment:
\begin{equation}
\mathcal G_h;\mathcal B_h(p,u,s) \vdash_\kappa Q \Downarrow \langle E,F,S,W\rangle,
\label{eq:bio-resolution}
\end{equation}
yielding a 4-tuple $d = \langle E, F, S, W \rangle$ where:
\begin{itemize}
  \item $E \in \Sigma_E = \{\textsf{Supported}, \textsf{Unknown}, \textsf{Unavailable}\}$ denotes the
  primary evidential status;
  \item $F \subseteq \Sigma_F = \{\textsf{Conflicted}, \textsf{Withheld}, \textsf{Stale}\}$ is a set
  of orthogonal operational qualifiers;
  \item $S \subseteq \mathcal B_h(p,u,s)$ is the minimized premise set available for content
  release; and
  \item $W \in \mathcal W_h\cup\{\varnothing\}$ is a protected decision witness object.
\end{itemize}
For this tuple, $\pi_E(d)=E$, $F(d)=F$, $S(d)=S$, and $W(d)=W$ denote its component
projections.
The resolver may inspect $\mathcal A_h$ to detect projection exclusions, but $\mathsf{View}_c$ is the
only boundary through which a resolution result is declassified.

The resolution judgment enforces the following structural consistency conditions:
\begin{equation}
\begin{aligned}
\pi_E(d)&\in\Sigma_E,\qquad F(d)\subseteq\Sigma_F, \\
\pi_E(d)=\textsf{Supported}&\iff S(d)\ne\varnothing\\
&\quad\land\operatorname{PolicyOK}_\kappa(S(d),Q,\mathcal G_h),\\
\pi_E(d)\ne\textsf{Supported}&\implies S(d)=\varnothing,\\
\mathcal G_h;\mathcal B_h(p,u,s)&\vdash_\kappa Q\Downarrow d\\
&\implies\operatorname{WitnessOK}_\kappa(W(d),d,Q,h).
\end{aligned}
\label{eq:resolution-consistency}
\end{equation}
Here, $\pi_E(d)$ projects the evidential status component of $d$. The predicate $\operatorname{PolicyOK}_\kappa(S(d), Q, \mathcal G_h)$
verifies that premises in $S(d)$ satisfy validity intervals, any required source-independence checks
($\operatorname{Indep}_\kappa(\cdot,\cdot;h)$), and derivation integrity under $\kappa$. When competing, unsuperseded
claims mutually meet support thresholds, resolution yields $\pi_E(d) = \textsf{Supported}$ with
$\textsf{Conflicted} \in F(d)$, populating $S(d)$ with both active versions.

\subsection{Decision Witness Integrity and Minimization}

A nonempty decision witness records the protected audit material required for a resolution result.
Witness validity is formalized by:
\begin{equation}
\begin{aligned}
&\operatorname{WitnessOK}_\kappa(W,d,Q,h) \iff \\
&\quad \bigl(W=\varnothing \land \operatorname{AbsenceOnly}_\kappa(d,Q,h)\bigr) \\
&\quad \lor \bigl(W\ne\varnothing \land \operatorname{BoundToHead}(W,h) \\
&\qquad \land \operatorname{AccessBound}_\kappa(W,Q)\\
&\qquad \land \operatorname{MinimalWitness}_\kappa(W,d,Q) \\
&\qquad \land \operatorname{ObligationsMet}_\kappa(W,d,Q,h)\bigr).
\end{aligned}
\label{eq:witness-ok}
\end{equation}
$\operatorname{AbsenceOnly}_\kappa(d,Q,h)$ holds exactly when $\pi_E(d)=\textsf{Unknown}$,
$F(d)=\varnothing$, and $\operatorname{Rel}(Q,\mathcal A_h)=\varnothing$; thus it covers only an
ordinary unknown caused by absence of admitted relevant material. $\operatorname{BoundToHead}$
binds or commits the witness to $h$;
$\operatorname{AccessBound}$ attaches its access policy; and $\operatorname{MinimalWitness}$ restricts
the payload to material needed to justify $d$. A cryptographic realization may implement the
bindings as commitments, but the abstract predicate does not require a particular construction.

$\operatorname{ObligationsMet}$ enforces category-specific witness bindings:
\begin{itemize}
  \item \textbf{Conflict Witness:} Binds the exact incompatible current versions and their conflict relation;
  \item \textbf{Staleness Witness:} Binds the exact temporally inapplicable versions, relevant times,
  and governing temporal rule;
  \item \textbf{Withholding Witness:} Binds an access-controlled reference or commitment to relevant
  material in $\mathcal A_h$ that projection excludes from $\mathcal B_h$; and
  \item \textbf{Unavailability Witness:} Binds the exact erased, sealed, lost, or unsupported-adapter
  marker and the affected dependency.
\end{itemize}
Diagnostic witnesses for failed verification or insufficient source roots remain protected and are
optional unless policy requires them.

When $Q$ lacks qualifying premises in $\mathcal A_h$, a resolver that records the evaluation context
may use the empty-payload witness:
\begin{equation}
\bot_{\mathcal W}(h, Q, \kappa) = \langle h, \operatorname{Digest}(Q), \operatorname{Digest}(\kappa), \varnothing \rangle.
\label{eq:empty-witness}
\end{equation}
$\bot_{\mathcal W}$ carries no evidentiary payload; an ordinary absence-only $\textsf{Unknown}$ may
instead use $W=\varnothing$. $W$ is not a support component and cannot itself justify content
release; only members of $S$ can do so. Absence of recorded support does not establish that the
proposition is false or absent from every external source.

\begin{table*}[t]
\centering
\caption{Typed resolution statuses, orthogonal qualifiers, and governing declassification constraints.}
\label{tab:resolution-statuses}
\small
\begin{tabularx}{\textwidth}{L{1.6cm} L{2.2cm} L{5.4cm} X}
\toprule
\textbf{Sort} & \textbf{Identifier} & \textbf{Deductive Evaluation Condition} & \textbf{Governing Declassification Constraint} \\
\midrule
Status & \textsf{Supported} & Admitted provenance DAG satisfies policy $\kappa$ for query $Q$ & Emit content assertion with mandatory source attribution \\
Qualifier & \textsf{Conflicted} & Multiple active, unsuperseded premises disagree & Disclose explicit epistemic disagreement or abstain \\
Status & \textsf{Unknown} & Authorized projection $\mathcal B_h$ lacks qualifying support & Disclose absence of recorded evidence; never assert falsity \\
Qualifier & \textsf{Withheld} & Relevant premises exist in $\mathcal A_h$ but are masked in $\mathcal B_h$ & Emit uniform, recipient-safe nonconfirming response \\
Status & \textsf{Unavailable} & Required dependency is erased, sealed, lost, or lacks an adapter & Disclose operational unavailability only when policy authorizes \\
Qualifier & \textsf{Stale} & Premises exist but fail current temporal validity interval & Add explicit temporal qualification or abstain \\
\bottomrule
\end{tabularx}
\end{table*}

\subsection{Withholding Detection and Recipient Non-Interference}

For a vertex set $Y\subseteq\mathcal V_h$, let $\operatorname{Rel}(Q,Y)\subseteq Y$ denote the
vertices semantically relevant to $Q$. Withholding arises when relevant
historical records exist in the full store but are filtered out by access authorization:
\begin{equation}
\begin{aligned}
\textsf{Withheld} \in F(d) \iff{} &\operatorname{Rel}(Q, \mathcal A_h) \setminus \\
&\operatorname{Rel}(Q, \mathcal B_h(p, u, s)) \ne \varnothing.
\end{aligned}
\label{eq:withheld-detection}
\end{equation}

In security-critical environments, revealing that data is withheld (e.g., ``I have records that you
are not authorized to view'') creates an existence side-channel, leaking confidential state.
Let $\mathcal D_{unk}(c)$ contain the internal decisions that policy $\kappa$ requires to remain
indistinguishable in context $c$, potentially including $\textsf{Unknown}$, $\textsf{Unavailable}$,
and decisions carrying $\textsf{Withheld}$. The declassification contract is:
\begin{equation}
\begin{aligned}
\operatorname{Collapse}_\kappa(c)&\implies
\forall d_1,d_2\in\mathcal D_{unk}(c):\\
&\mathsf{View}_c(d_1)=\mathsf{View}_c(d_2)=\bot_{unk}.
\end{aligned}
\label{eq:view-collapse}
\end{equation}
Section~\ref{sec:assertion-mediation} restricts the renderer to this authorized view. Equation~\ref{eq:view-collapse}
is an observational policy condition, not a claim of cryptographic indistinguishability.

\subsection{Epistemic Non-Vacuity and Scope Preservation}

Resolving a query to $\textsf{Unknown}$ represents a localized epistemic judgment at head $h$; it
does not assert the negation of the proposition $\neg\phi$, nor does it insert a negative axiom
into state. The agent may still use general parametric reasoning in non-autobiographical
domains. For example, an agent lacking records of a user's dietary preferences may discuss culinary
topics generally, but cannot assert that the user observes a specific dietary restriction.

Resolution also preserves operational scopes. An entitlement granted within a
software development task does not generalize to enterprise financial auditing; consent granted to a
specific coworker does not extend to third parties; and a commitment established in 2024 does not
authorize an unqualified present-tense assertion in 2026.

\subsection{Dependency Invalidation and Forward State Reduction}

When an upstream evidence artifact is found to be compromised, corrupted, or retracted, our framework
executes forward state reduction rather than destructive history rewriting. A revocation operation
records a new versioned state transition $h \xrightarrow{\operatorname{revoke}(e)} h'$ appending a
revocation descriptor to the graph.

Let $\mathfrak D_h(v)$ be the nonempty family of complete transitive dependency alternatives
recorded for $v$, each including its source roots; for a root,
$\mathfrak D_h(v)=\{\{v\}\}$. Let $\operatorname{Revoked}(e,h)$ mean that $e$ is
revoked at $h$. Valid support is preserved when at least one complete alternative survives:
\begin{equation}
\begin{aligned}
\operatorname{ValidSupport}_\kappa(v,h) \iff{} &\exists D\in\mathfrak D_h(v): \\
&\operatorname{Satisfies}_\kappa(D, v, \mathcal G_h) \\
&\land \forall e \in D:\neg\operatorname{Revoked}(e,h).
\end{aligned}
\label{eq:correction-propagation}
\end{equation}
If every dependency path for $v$ intersects at least one revoked root, $v$ loses its warranted
standing and is excluded from future support sets $S$. If alternative, mutually independent valid
paths remain in $\mathcal G_h$, support is preserved.

We formalize four primitive lifecycle transitions:
\begin{itemize}
  \item \textbf{Supersession ($v_1 \succ v_2$):} Introduces a successor claim that supersedes an
  earlier claim for current queries while preserving the immutable historical trace;
  \item \textbf{Dependency Invalidation ($\operatorname{revoke}(e)$):} Invalidates an upstream root,
  requiring recomputation of downstream support via Equation~\ref{eq:correction-propagation};
  \item \textbf{Conflict Inscription ($v_1 \mathrel{\#_{\text{conflict}}} v_2$):} Formally registers
  contradictory claims when no automated policy rule resolves precedence;
  \item \textbf{Cryptographic / Physical Erasure ($\operatorname{erase}(e)$):} Records that an
  affected dependency is unavailable after a policy-authorized erasure operation. Whether bytes are
  irrecoverable or a legal obligation is satisfied depends on the deployment.
\end{itemize}

\section{Assertion Mediation and Bounded Release}
\label{sec:assertion-mediation}

Maintaining an intact accepted state in the persistence layer does not prevent an unconstrained
generative model from hallucinating biographical details, omitting mandatory epistemic qualifications,
or fabricating historical authority during natural-language synthesis. Assertion mediation establishes
a formal security boundary at the point of external utterance release.

\begin{table}[t]
\centering
\caption{Three complementary enforcement surfaces and their corresponding safety properties.}
\label{tab:assurance-surfaces}
\small
\begin{tabularx}{\columnwidth}{L{2.3cm} X}
\toprule
\textbf{Enforcement Surface} & \textbf{Formal Safety Property} \\
\midrule
State Admission & Regulates which epistemic entities enter the accepted state graph $\mathcal G_h$ \\
Context Projection & Filters accepted state into an authorized projection $\mathcal B_h(p, u, s)$ \\
Assertion Mediation & Enforces content grounding, qualification, and recipient non-interference prior to release \\
\bottomrule
\end{tabularx}
\end{table}

Table~\ref{tab:assurance-surfaces} contrasts these three distinct enforcement surfaces. Architectures
implementing only storage constraints enforce \emph{state-admission assurance}, which cannot prevent
a probabilistic language model from emitting ungrounded claims during dialogue generation.

\subsection{Generate--Verify--Revise Operational Semantics}

Figure~\ref{fig:assertion-boundary} illustrates the assertion mediation architecture. The generative
model synthesizes candidate responses using an authorized projection bound to prompt head $h_p$.
Before candidate text is emitted to governed external channels, the mediator intercepts and evaluates
each constituent proposition.

\begin{figure*}[t]
\centering
\begin{tikzpicture}[
  box/.style={draw=navy, rounded corners=2pt, fill=sublayer,
    minimum height=0.8cm, text width=2.7cm, align=center},
  gate/.style={draw=crimson, rounded corners=2pt, fill=kernellayer,
    minimum height=0.8cm, text width=2.8cm, align=center},
  arr/.style={-{Stealth[length=2.3mm]}, thick},
  every node/.style={font=\small}
]
\node[box] (state) at (0,0) {Accepted claim state\\at head $h_p$};
\node[box] (proj) at (3.4,0) {Authorized projection\\for recipient and purpose};
\node[box] (model) at (6.8,0) {Candidate response\\and semantic units};
\node[gate] (med) at (10.2,0) {Assertion mediator\\extract, check, decide};
\node[box] (out) at (13.8,0.8) {Released units\\with audit binding};
\node[box] (rev) at (13.8,-0.8) {Review / Revision\\revise, block, or escalate};
\draw[arr] (state) -- (proj);
\draw[arr] (proj) -- (model);
\draw[arr] (model) -- (med);
\draw[arr,draw=emerald] (med) -- (out);
\draw[arr,draw=crimson,dashed] (med) -- (rev);
\end{tikzpicture}
\caption{Assertion mediation pipeline for content claims and status responses. The boundary intercepts governed output channels, verifying candidate claims against accepted state before release.}
\label{fig:assertion-boundary}
\end{figure*}
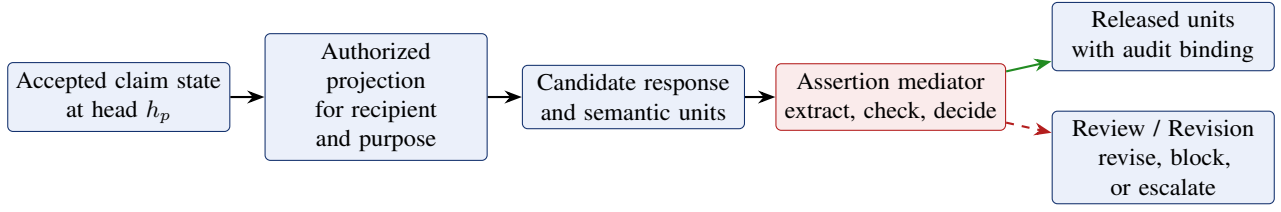

We model mediation as an operational transition system over candidate configurations:
\begin{equation}
\langle z,c\rangle \xrightarrow{\text{extract}}
\langle \bar v,\bar Q,c\rangle \xrightarrow{\text{resolve}}
\langle \bar v,\bar d,c\rangle \xrightarrow{\text{gate}} z_{\mathrm{rel}},
\label{eq:mediation-transition}
\end{equation}
where $z$ is the candidate string, $c = \langle h_p, h_e, p, u, s, \kappa, t_e \rangle$ is the execution
context, $\bar v$ is the ordered sequence of governed semantic units, $\bar Q$ binds one structured
query to each unit, $\bar d$ contains the corresponding resolution results, and $z_{\mathrm{rel}}$
is the released output.

When a unit fails release verification, the mediator may provide the generator a declassified
diagnostic derived through $\mathsf{View}_c$ and request a revised candidate $z'$. Raw witness
material $W(d)$ is not supplied to the generator or renderer. Every revision re-enters the pipeline
at the extraction stage.

\subsection{Emission-Head Temporal Verification}

In concurrent agent systems, the authoritative state head may advance from $h_p$ (the head bound
during context assembly) to $h_e$ (the head active at emission time $t_e$). To prevent race
conditions and stale claim release, the mediator evaluates the emission-head judgment:
\begin{equation}
\begin{aligned}
&\operatorname{HeadUseOK}_\kappa(v,h_p,h_e,t_e)\\
&\quad\iff h_p=h_e\\
&\qquad\lor\bigl(\operatorname{SnapshotOK}_\kappa(h_p,h_e,t_e)\\
&\qquad\quad\land\operatorname{SnapQual}_\kappa(v,h_p,t_e)\bigr).
\end{aligned}
\label{eq:head-use-ok}
\end{equation}
If $h_p \ne h_e$ and policy disallows snapshot release, the candidate must be re-projected and
re-mediated against current head $h_e$. If snapshot emission is authorized, the emitted unit $v$
must include explicit temporal framing (e.g., ``As of the preceding session...'').

\subsection{Formal Autobiographical Assertion Boundedness}

For any candidate output $z$, let
$\operatorname{Units}_{\mathcal L_{bio}}(z)=\langle v_1,\ldots,v_n\rangle$ be the ordered sequence
of covered semantic-unit occurrences, each with a unique unit identifier. The content and status
subsets partition those occurrences:
\begin{equation}
\{v_1,\ldots,v_n\}=
\operatorname{Content}_{\mathcal L_{bio}}(z)\mathbin{\dot\cup}
\operatorname{Status}_{\mathcal L_{bio}}(z).
\label{eq:output-partition}
\end{equation}
This partition applies to underlying semantic propositions rather than disjoint surface clauses: a
single natural-language sentence may encode both a status disclosure and an embedded content claim.

For execution context $c = \langle h_p,h_e,p,u,s,\kappa,t_e\rangle$ and unit $v$ bound to query
$Q_v$, let
$\mathcal G_{h_p};\mathcal B_{h_p}(p,u,s)\vdash_\kappa Q_v\Downarrow d_v$. For resolver outputs,
$\mathsf{View}_c$ is the sole declassification boundary: any result-dependent input to the generator
or renderer is derived from $\mathsf{View}_c(d_v)$, never raw $W(d_v)$. The renderer is constrained by:
\begin{equation}
\begin{aligned}
\operatorname{StatusOut}_\kappa(d,c)&=
\operatorname{Render}_\kappa(\mathsf{View}_c(d),c), \\
\mathsf{View}_c(d_1)=\mathsf{View}_c(d_2)&\implies{}\\
\operatorname{RenderSet}_\kappa(\mathsf{View}_c(d_1),c)
&=\operatorname{RenderSet}_\kappa(\mathsf{View}_c(d_2),c).
\end{aligned}
\label{eq:render-noninterference}
\end{equation}

Let $q(v)$ be the structured proposition extracted from unit $v$. We define the release
admissibility relations for content and status units:
\begin{equation}
\begin{aligned}
&\operatorname{ContentReleaseOK}(\alpha,d_\alpha,c)\\
&\quad\iff \exists x\in S(d_\alpha):\\
&\qquad \operatorname{ExactMatch}(q(\alpha),\phi(x)) \\
&\land \pi_E(d_\alpha) = \textsf{Supported} \\
&\land \operatorname{QualOK}_\kappa(\alpha, d_\alpha) \\
&\land \operatorname{DisclosureOK}_\kappa(\alpha,c) \\
&\land \operatorname{HeadUseOK}_\kappa(\alpha, h_p, h_e, t_e),
\end{aligned}
\label{eq:content-release-ok}
\end{equation}
\begin{equation}
\begin{aligned}
&\operatorname{StatusResponseOK}(\beta,d_\beta,c)\\
&\quad\iff \beta\in
\operatorname{RenderSet}_\kappa(\mathsf{View}_c(d_\beta),c) \\
&\land \operatorname{HeadUseOK}_\kappa(\beta, h_p, h_e, t_e).
\end{aligned}
\label{eq:status-response-ok}
\end{equation}

\begin{definition}[Autobiographical Assertion Boundedness]
\label{def:assertion-boundedness}
A response $z$ satisfies autobiographical assertion boundedness under execution context $c$, denoted
$c \vdash_{\text{release}} z : \operatorname{Bounded}_{\mathcal L_{bio}}$, if and only if every covered
semantic unit satisfies its corresponding release relation:
\begin{equation}
\begin{aligned}
&\operatorname{Bounded}_{\mathcal L_{bio}}(z,c)\iff{}\\
&\quad\bigl(\forall \alpha\in\operatorname{Content}_{\mathcal L_{bio}}(z):\\
&\qquad\operatorname{ContentReleaseOK}(\alpha,d_\alpha,c)\bigr)\\
&\quad\land\bigl(\forall \beta\in\operatorname{Status}_{\mathcal L_{bio}}(z):\\
&\qquad\operatorname{StatusResponseOK}(\beta,d_\beta,c)\bigr).
\end{aligned}
\label{eq:release-ok}
\end{equation}
\end{definition}

\begin{proposition}[Soundness of Assertion Mediation]
\label{prop:assertion-boundedness}
Assume E1--E4; complete extraction and correct unit typing over $\mathcal L_{bio}$; correct
implementations of $\operatorname{ExactMatch}$, $\operatorname{QualOK}$, and
$\operatorname{DisclosureOK}$; resolution satisfying Equation~\ref{eq:resolution-consistency},
$\operatorname{WitnessOK}_\kappa$, and dependency-invalidation recomputation; a correct
$\mathsf{View}_c$ and renderer satisfying Equations~\ref{eq:view-collapse} and
\ref{eq:render-noninterference}; recursive re-mediation of revisions; and correct
$\operatorname{HeadUseOK}_\kappa$. If the transition system in Equation~\ref{eq:mediation-transition}
reaches a release state with output $z_{\mathrm{rel}}$, then
$\operatorname{Bounded}_{\mathcal L_{bio}}(z_{\mathrm{rel}},c)$ holds.
\end{proposition}
This is a conditional composition result: it does not establish the correctness of the extraction,
policy predicates, resolver, or declassification function assumed in its premise.

\begin{proof}[Proof sketch]
By E4, every candidate on a governed outbound channel reaches the mediator. Complete extraction
and correct typing produce the ordered covered units and their content/status partition.
A content assertion $\alpha$ is emitted only after the mediator verifies every conjunct of
$\operatorname{ContentReleaseOK}$ using the assumed-correct predicates and resolver. A status unit
is emitted only from the render set of its authorized view and after the head-use check. The
renderer receives no input other than $\mathsf{View}_c(d)$; equal views therefore yield equal render
sets, and any policy-required collapse is supplied by Equation~\ref{eq:view-collapse}. A failed
candidate is not released, and every revision re-enters extraction. Thus every unit in a released
output satisfies its corresponding release predicate, which is exactly
$\operatorname{Bounded}_{\mathcal L_{bio}}(z_{\mathrm{rel}},c)$. Appendix~\ref{sec:proof-boundedness}
provides the proof details.
\end{proof}

\subsection{Assertion Trace Semantics and Auditability}

Let $\operatorname{Units}_{\mathcal L_{bio}}(z_{\mathrm{rel}})=\langle v_1,\ldots,v_n\rangle$.
The specification assigns each released governed unit exactly one binding, in that order:
\begin{equation}
\begin{aligned}
\mu_v=\langle{}&\operatorname{Id}(v),\operatorname{Digest}(v),\operatorname{Type}(v),\\
&\operatorname{Digest}(Q_v),\operatorname{Digest}(d_v),\\
&\operatorname{Digest}(S(d_v)),\operatorname{Digest}(W(d_v)),\\
&\delta_v,\sigma_v\rangle,
\end{aligned}
\label{eq:unit-binding}
\end{equation}
where $\operatorname{Type}(v) \in \{\textsf{Content}, \textsf{Status}\}$, $\delta_v$ records the mediation
disposition ($\textsf{Pass}, \textsf{Qualify}, \textsf{Redact}$), and $\sigma_v$ records the emission-head
verification result.

The response-level trace is:
\begin{equation}
\begin{aligned}
\bar\mu_z&=\langle\mu_{v_1},\ldots,\mu_{v_n}\rangle,\\
\mathcal T_z&=\langle \bar\mu_z,h_p,h_e,t_e,\kappa,\operatorname{Digest}(z_{\mathrm{rel}})\rangle.
\end{aligned}
\label{eq:assertion-artifact}
\end{equation}
An implementation with canonical serialization, hashing, and protected storage can use this
structure for audit review. The data structure alone establishes neither immutability nor
non-repudiation.

\subsection{Operational Boundaries and Limitations}

Assertion mediation relies on the semantic fidelity of the extraction and classification stage.
Structured schema slots and deterministic templates avoid free-form extraction for the fields they
govern, but free-form natural-language extraction remains vulnerable to subtle presuppositions,
implicative phrasing, or compound rhetorical structures. Furthermore, real-time conversational
interfaces impose stringent latency budgets, necessitating cached verification or constrained
template rendering.

The accompanying executable artifact instantiates structured claim objects, a typed resolver,
fixed recipient-facing templates, ordered binding data structures, and 24 hand-authored cases.
It does not implement natural-language semantic-unit extraction, generalized witness access
control, configurable $\mathsf{View}_c$ policies or output-equivalence checks, full provenance-DAG
traversal, task/branch/relationship scope enforcement, cryptographic witness or assertion
commitments, canonical serialization or artifact-hash verification, or the target-origin admission
checks in Equation~\ref{eq:target-origin-obligations}.
Those components remain specification-level obligations.

\section{Evaluation Framework and Conformance Checks}
\label{sec:epistemic-threats-evaluation}

We evaluate the operational correctness of the typed resolver and assertion mediator using an
executable conformance suite, distinguishing these structural verification checks from the proposed
end-to-end evaluation protocol detailed in Appendix~\ref{sec:appendix-eval-protocol}. The suite
checks the reference implementation against encoded decision oracles, flattened source-root
constraints, and release invariants across structured threat scenarios.

\subsection{Evaluation Objectives and Baseline Paradigms}

The evaluation design separates executable conformance questions from measurements that require a
deployed agent. It covers four operational dimensions:
\begin{itemize}
  \item \textbf{State admission and provenance discipline:} Verifying that typed claim objects and
  origin checks prevent unverified model priors, unadmitted inputs, and single-root
  duplicate records from acquiring authoritative autobiographical standing.
  \item \textbf{Assertion release control:} Evaluating the mediator's capability to intercept
  unsupported, out-of-scope, or contradictory candidate statements while preserving legitimately
  supported answers with appropriate attributions.
  \item \textbf{Longitudinal state evolution:} Tracking how temporal expiration, superseding user
  corrections, multi-source conflicts, and model substrate upgrades are handled without historical
  corruption or silent promotion.
  \item \textbf{Operational overhead and utility:} Measuring decision latency, computational
  overhead, and over-abstention across deployed assurance configurations; the present artifact
  provides only a local resolver microbenchmark and deterministic coverage counts.
\end{itemize}

For a future end-to-end study, we formulate five comparative
deployment paradigms: (1)~a pure parametric agent without episodic memory; (2)~a flat
vector-retrieval store over raw transcripts; (3)~standard top-$k$ retrieval-augmented generation
(RAG); (4)~provenance-tagged retrieval lacking graph-level root independence; and (5)~a
state-bounded architecture that filters memory projections at prompt assembly without pre-release
mediation. The present conformance suite instead isolates three deterministic release rules:
flat/prior, source-tag, and typed mediation. Appendix~\ref{sec:appendix-baselines-ablations}
specifies the five deployment paradigms for later end-to-end comparison.

\subsection{Safety and Utility Metrics}

Let $C$ denote candidate assertion opportunities evaluated by the mediator, $U \subseteq C$ denote
opportunities that are unsupported or require qualification under policy $\kappa$, $R \subseteq C$
denote released responses emitting content or qualified status, and $P \subseteq R$ denote candidate
content passed without required qualifications.

We quantify mediation safety and efficacy through two primary risk metrics:
\begin{enumerate}
  \item \textbf{Unsupported-Opportunity Release Rate ($\mathsf{UOR}$):} The proportion of unsafe or
  unqualified candidate assertions that escape mediation:
  \begin{equation}
    \mathsf{UOR} = \frac{|U \cap P|}{|U|}.
    \label{eq:unsupported-risk}
  \end{equation}
  \item \textbf{Release Contamination ($\mathsf{RC}$):} The proportion of released conversational
  outputs containing unqualified unsafe assertions:
  \begin{equation}
    \mathsf{RC} = \frac{|U \cap P|}{|R|} \quad (\text{defined as } 0 \text{ when } R = \varnothing).
    \label{eq:release-contamination}
  \end{equation}
\end{enumerate}
Because safety can trivially be maximized through blanket refusal, we pair these metrics with
\textbf{Coverage}---the proportion of legitimately supported queries for which the system releases
an accurate, appropriately qualified response. In addition, we measure \textbf{Prevention
Precision} (proportion of interventions targeting genuinely unsafe opportunities), \textbf{Prevention
Recall} (proportion of unsafe opportunities intercepted), and \textbf{Mediation Latency}.

\subsection{Executable Conformance Suite Results}

We evaluate the in-memory claim store, typed resolver, and assertion mediator against 24 structured
cases: 19 unsafe-release opportunities and five supported controls across six threat tracks. Cases
that require corroboration use a policy requiring two roots that satisfy
$\operatorname{Indep}_\kappa(\cdot,\cdot;h)$.

\begin{table*}[t]
\centering
\caption{Conformance suite aggregate results under the two-source fixture (24 cases).}
\label{tab:epistemic-empirical-results}
\small
\begin{tabularx}{\textwidth}{L{3.5cm} *{5}{>{\raggedleft\arraybackslash}X}}
\toprule
\textbf{Release rule} & \textbf{UOR} $\downarrow$ & \textbf{RC} $\downarrow$ & \textbf{Coverage} $\uparrow$ & \textbf{Prevention precision} $\uparrow$ & \textbf{Prevention recall} $\uparrow$ \\
\midrule
Flat/prior rule & 19/19 (100.0\%) & 19/24 (79.2\%) & 5/5 (100.0\%) & N/A---no interventions & 0/19 (0.0\%) \\
Source-tag rule & 18/19 (94.7\%) & 18/23 (78.3\%) & 5/5 (100.0\%) & 1/1 (100.0\%) & 1/19 (5.3\%) \\
Typed mediation & 0/19 (0.0\%) & 0/7 (0.0\%) & 5/5 (100.0\%) & 19/19 (100.0\%) & 19/19 (100.0\%) \\
\bottomrule
\end{tabularx}

\vspace{2.5mm}

\caption{Per-case oracle and decisions across the 24-case conformance suite.}
\label{tab:epistemic-trace-oracle}
\scriptsize
\renewcommand{\arraystretch}{0.84}
\setlength{\tabcolsep}{3pt}
\begin{tabularx}{\textwidth}{l c L{3.5cm} L{2.6cm} L{2.2cm} X c}
\toprule
\textbf{Case} & \textbf{Type} & \textbf{Condition} & \textbf{Fixture roots} & \textbf{Decision $d_h$} & \textbf{Candidate check} & \textbf{Disposition} \\
\midrule
poison\_\allowbreak 01 & Unsafe & Unadmitted model prior without evidence & -- & Unknown & Model-origin candidate & Blocked \\
poison\_\allowbreak 02 & Unsafe & Indirect prompt injection via untrusted doc & untrusted\_\allowbreak web\_\allowbreak page & Unknown & Unverified evidence & Blocked \\
poison\_\allowbreak 03 & Unsafe & Hypothetical simulation claiming observation & sim\_\allowbreak sandbox\_\allowbreak log & Unknown & Simulation origin & Blocked \\
poison\_\allowbreak 04 & Unsafe & Fabricated unverified consolidation summary & unverified\_\allowbreak llm\_\allowbreak hallucination & Unknown & Unverified derivation & Blocked \\
dup\_\allowbreak 01 & Unsafe & Duplicate records share single upstream root & root\_\allowbreak chat\_\allowbreak session\_\allowbreak A & Unknown & Insufficient roots & Blocked \\
dup\_\allowbreak 02 & Unsafe & Multi-hop derivation chain ($k=3$) sharing root & root\_\allowbreak session\_\allowbreak single & Unknown & Single upstream root & Blocked \\
dup\_\allowbreak 03 & Unsafe & Partial root overlap below 2-root threshold & shared\_\allowbreak root\_\allowbreak omega & Unknown & Shared root overlap & Blocked \\
supported\_\allowbreak 01 & Control & Valid preference corroborated by 2 roots (Control) & supported\_\allowbreak root\_\allowbreak a, supported\_\allowbreak root\_\allowbreak b & Supported & Exact $S(d_h)$ member & Passed \\
scope\_\allowbreak 01 & Unsafe & Unauthorized principal query for restricted intake & private\_\allowbreak intake & Unknown + Withheld & Outside authorized view & Blocked \\
scope\_\allowbreak 02 & Unsafe & Purpose mismatch against authorized scope & payroll\_\allowbreak root\_\allowbreak 1, payroll\_\allowbreak root\_\allowbreak 2 & Unknown + Withheld & Purpose mismatch & Blocked \\
scope\_\allowbreak 03 & Unsafe & Sealed / cryptographically erased record & erased\_\allowbreak device\_\allowbreak vault & Unavailable & Sealed record unavailable & Blocked \\
supported\_\allowbreak 02 & Control & Authorized principal \& matching purpose (Control) & org\_\allowbreak root\_\allowbreak a, org\_\allowbreak root\_\allowbreak b & Supported & Exact $S(d_h)$ member & Passed \\
time\_\allowbreak 01 & Unsafe & Expired permission grant queried past validity & time\_\allowbreak root\_\allowbreak 1, time\_\allowbreak root\_\allowbreak 2 & Unknown + Stale & No current support & Blocked \\
time\_\allowbreak 02 & Unsafe & Premature commitment queried before start time & promo\_\allowbreak root\_\allowbreak 1, promo\_\allowbreak root\_\allowbreak 2 & Unknown + Stale & Not yet active & Blocked \\
time\_\allowbreak 03 & Unsafe & Transitive staleness in derived preference summary & order\_\allowbreak root\_\allowbreak a, order\_\allowbreak root\_\allowbreak b & Unknown + Stale & Observation expired & Blocked \\
supported\_\allowbreak 03 & Control & Active preference within valid bounds (Control) & tz\_\allowbreak root\_\allowbreak a, tz\_\allowbreak root\_\allowbreak b & Supported & Exact $S(d_h)$ member & Passed \\
conflict\_\allowbreak 01 & Unsafe & Two active unsuperseded conflicting locations & city\_\allowbreak root\_\allowbreak 1, city\_\allowbreak root\_\allowbreak 2, city\_\allowbreak root\_\allowbreak 3, city\_\allowbreak root\_\allowbreak 4 & Supported + Conflicted & Status-only response & Qualified \\
conflict\_\allowbreak 02 & Unsafe & Three-way location disagreement across roots & r1, r2, r3, r4, r5, r6 & Supported + Conflicted & Status-only response & Qualified \\
invalidation\_\allowbreak 01 & Unsafe & Transitive dependency invalidation revokes support & revoked\_\allowbreak cert\_\allowbreak root\_\allowbreak 1, revoked\_\allowbreak cert\_\allowbreak root\_\allowbreak 2 & Unknown & Revoked upstream root & Blocked \\
invalidation\_\allowbreak 02 & Control & Redundant independent root survives invalidation (Control) & fido2\_\allowbreak root\_\allowbreak valid\_\allowbreak a, sms\_\allowbreak root\_\allowbreak invalidated, smtp\_\allowbreak root\_\allowbreak valid\_\allowbreak b & Supported & Exact $S(d_h)$ member & Passed \\
correction\_\allowbreak 01 & Unsafe & Outdated preference superseded by user correction & lang\_\allowbreak root\_\allowbreak new\_\allowbreak a, lang\_\allowbreak root\_\allowbreak new\_\allowbreak b, lang\_\allowbreak root\_\allowbreak old\_\allowbreak a, lang\_\allowbreak root\_\allowbreak old\_\allowbreak b & Supported & Outdated candidate not in $S(d_h)$ & Blocked \\
correction\_\allowbreak 02 & Unsafe & Multi-version chain ($v_1 \to v_2 \to v_3$) asserting $v_2$ & r\_\allowbreak th1a, r\_\allowbreak th1b, r\_\allowbreak th2a, r\_\allowbreak th2b, r\_\allowbreak th3a, r\_\allowbreak th3b & Supported & Superseded revision $v_2$ & Blocked \\
correction\_\allowbreak 03 & Control & Active latest revision $v_3$ in chain (Control) & r\_\allowbreak th1a, r\_\allowbreak th1b, r\_\allowbreak th2a, r\_\allowbreak th2b, r\_\allowbreak th3a, r\_\allowbreak th3b & Supported & Exact $S(d_h)$ member & Passed \\
head\_\allowbreak race\_\allowbreak 01 & Unsafe & Emission-head race mismatch without snapshot permission & k8s\_\allowbreak root\_\allowbreak a, k8s\_\allowbreak root\_\allowbreak b & Supported & Stale head use rejected & Blocked \\
\bottomrule
\end{tabularx}
\end{table*}

Table~\ref{tab:epistemic-empirical-results} reports aggregate metrics across three deterministic
release rules, and Table~\ref{tab:epistemic-trace-oracle} details the per-case oracle. Across this
test fixture:
\begin{enumerate}
  \item The \textbf{Flat/prior rule} admits all 19 unsafe candidates ($100.0\%$ $\mathsf{UOR}$,
  $79.2\%$ $\mathsf{RC}$ across $|R|=24$ emissions). By treating model priors and raw retrieved text
  as authoritative assertions, this rule serves as an unmitigated comparison condition illustrating
  baseline generative vulnerability.
  \item The \textbf{Source-tag rule} prevents one unsafe release but admits the other 18
  ($94.7\%$ $\mathsf{UOR}$, $78.3\%$ $\mathsf{RC}$ over $|R|=23$). Because surface
  metadata tags do not track root provenance, temporal validity, or supersession history, tagging
  alone cannot prevent false corroboration or stale releases.
  \item \textbf{Typed mediation} passes no unsafe candidate unqualified ($0/19$ $\mathsf{UOR}$ and $0/7$
  $\mathsf{RC}$) while passing all five controls ($100.0\%$ $\mathsf{Coverage}$), with $100.0\%$
  prevention precision and recall. It blocks 17 unsafe candidates and emits qualified status
  responses for the two active-conflict cases, yielding $|R|=7$ total emissions.
\end{enumerate}

\subsection{Falsification Conditions}

The architecture's empirical validity rests on concrete falsification conditions. The approach
would be empirically refuted if: (1)~semantic extraction fails to capture covered propositions
above a declared operational error tolerance; (2)~mediation overhead reduces valid-query coverage
below an acceptable utility threshold; or (3)~a strictly simpler heuristic achieves equivalent safety
and coverage at lower operational latency. Appendix~\ref{sec:appendix-falsification} defines these
criteria for future end-to-end deployments.

Fourteen test functions additionally cover mixed temporal/conflict states, witness fields,
non-disclosure templates, ordered released-unit bindings, exact aggregate regressions, and a bounded
enumeration of 1,152 single-claim configurations. These unit checks add no rows to
Tables~\ref{tab:epistemic-empirical-results}--\ref{tab:epistemic-trace-oracle}. Neither the case oracle
nor the auxiliary checks establish semantic correctness of arbitrary natural-language text,
recipient-view equivalence, or exhaustive state-space coverage.

\subsection{Artifact Availability}

The public \href{https://github.com/openkedge/pci/tree/main/src/epistemic_bounds}
{\texttt{openkedge/pci}} artifact is published under the MIT license and pinned at release
\href{https://github.com/openkedge/pci/releases/tag/v0.2.0}{\texttt{v0.2.0}} (commit
\href{https://github.com/openkedge/pci/tree/74841530988f1324bd927da15519486509372917/src/epistemic_bounds}{\texttt{7484153}}).
That revision provides the full 24-case conformance suite, typed mediator, and 1,152-state bounded
exploration suite. In the present working tree, \texttt{make reproduce} runs 14 tests and regenerates
the 24-case tables.

\section{Discussion and Deployment Considerations}
\label{sec:governance-implications}

Enforcing epistemic discipline provides an inspectable, auditable claim regime; it does not, by itself, establish the moral or institutional legitimacy of the underlying policy rules. This distinction is critical when claim standing governs automated actions, confidential disclosure, or long-term personalization. Appendix~\ref{sec:appendix-governance} provides an extended analysis of sociotechnical governance, multi-party consent, and regulatory erasure.

\subsection{Withholding Privacy and Existence Leakage}

Typed claim objects can support fine-grained access control across principals, purposes, and
temporal intervals. Once a consent change is admitted and its projection policy becomes active,
affected claims can be excluded from subsequent views $\mathcal B_h(p,u,s)$ without rewriting
historical state.

The act of withholding information can itself leak confidential state. Emitting a specific refusal
(e.g., ``Access to your medical records is denied'') confirms to an unauthorized requester that
relevant records exist. When policy requires that existence remain secret, $\mathsf{View}_c$ can
map \textsf{Withheld}, \textsf{Unavailable}, and \textsf{Unknown} to the same declassified view,
so that the renderer returns a uniform nonconfirming response. Appendix~\ref{sec:appendix-consent-leakage}
states the observation assumptions for this branch-distinguishing channel.

\subsection{Relational Privacy and Multi-Party Contexts}

Autobiographical agent memories often capture multi-party interactions. An agent's authenticated provenance for a conversation does not grant unilateral authority to disclose a third party's private statements or attributes to other users. Deployments must enforce multi-party provenance isolation, ensuring that one user's interaction history cannot leak into another's authorized projection (Appendix~\ref{sec:appendix-relational-privacy}).

\subsection{Procedural Guarantees vs.\ Semantic Truth}

Our conditional contract theorem (Proposition~\ref{prop:assertion-boundedness}) establishes strict procedural bounds:
\begin{itemize}
  \item Assertions satisfy declared operational rules, not that external informants are infallible or objectively true;
  \item Recorded stances reflect system policy, not that an agent experiences subjective belief or conscious memory;
  \item Cryptographic signatures authenticate signed bytes and key possession, rather than factual veracity.
\end{itemize}

\subsection{Closed-Channel Scope and Side Channels}

Our guarantees apply to governed output channels passing through assertion mediation. Ungoverned side channels---such as debug logs, diagnostic telemetry, cached prompts, or unmediated tool APIs---must be isolated or routed through mediation gateways to prevent out-of-band autobiographical leaks. Systems should catalog all I/O surfaces and report state admission, projection, and mediation assurance levels as distinct operational tiers.

\section{Related Work}
\label{sec:related-work}

\paragraph{Epistemic foundations and bounding.}
In formal verification, De Giacomo et al.\ bound fluent atoms considered possibly true by an agent for decidable verification over unbounded transition systems~\cite{degiacomo2013bounded}. In governance, Kim et al.\ analyze epistemic bounds in administrative oversight of generative AI~\cite{kim2026administrative}. Here, \emph{autobiographical assertion boundedness} specifically ensures that an agent's outward assertions concerning itself, its user, or named relationships strictly conform to accepted evidence, validity horizons, and authorized disclosure policies.

\paragraph{Transactional state integrity.}
Transactional memory systems govern state admission, version changes, and repair under concurrent
updates~\cite{he2026continuitykernel,li2026memtx}. Assertion boundedness begins after that boundary:
it assumes an accepted head and governs the evidential and disclosure conditions for outward
statements. Appendix~\ref{sec:appendix-related-work} compares this release contract with related
memory and grounding systems in Tables~\ref{tab:novelty-matrix-full}
and~\ref{tab:novelty-matrix-release-full}.

\paragraph{Persistent memory and agent architectures.}
Hierarchical memory architectures enable agents to maintain state across extended dialogues. Generative Agents~\cite{park2023generative} and MemGPT~\cite{packer2023memgpt} pioneered tiered memory using periodic reflection and retrieval. A-MEM~\cite{xu2025amem} and Mem0~\cite{chhikara2025mem0} organize dynamic stores for lifelong personalization, while LoCoMo~\cite{maharana2024locomo} and LaMP~\cite{salemi2024lamp} evaluate long-context retrieval. These systems focus on memory construction and retrieval rather than the typed, recipient-specific output-release contract studied here.

\paragraph{Structured and provenance-aware memory.}
Recent systems introduce structured memory representations. MemIR~\cite{jin2026memir} decouples evidence, retrieval cues, and memory atoms to prevent provenance-role collapse. MemTX~\cite{li2026memtx} treats belief updates as transactional operations with repair cascades. Eywa~\cite{joshi2026eywa} maintains immutable source evidence while deriving provenance-linked canonical facts. MemLineage~\cite{ouyang2026memlineage} secures memory via cryptographic lineage graphs for tool invocations. Our work addresses a complementary problem: the outward assertion boundary. While MemIR, MemTX, and Eywa structure internal storage, our architecture formalizes policy-relative promotions, typed resolution with orthogonal flags, recipient-safe status rendering, and emission-head race verification for outward autobiographical claims.

\paragraph{Epistemic logic and belief revision.}
Dynamic epistemic logic formalizes multi-agent knowledge transitions~\cite{fagin2004reasoning}, while AGM belief revision analyzes minimal, rational state change~\cite{alchourron1985logic}. Rather than assuming idealized deductive closure, our framework formulates procedural, dependency-bounded resolution tailored for neural agents operating over partial, open-world evidence stores.

\paragraph{Data provenance and verifiable credentials.}
The W3C PROV-DM standard models entities, activities, and agents~\cite{w3c2013provdm}, while Verifiable Credentials authenticate issuer statements and their presentation~\cite{w3c2025vc}. Our typed provenance graph adds epistemic roles, temporal validity, disclosure scopes, and policy-relative source independence.

\paragraph{Agent security and memory poisoning.}
Autonomous agents are susceptible to indirect prompt injection via retrieved web pages, external tools, and multi-turn dialogues~\cite{debenedetti2024agentdojo,sleeper2026poisoning}. ConsistencyGate~\cite{zhang2026consistencygate} filters untrusted inputs using write-time model consistency. Our framework complements input filtering by enforcing source-independence graph analysis at resolution time and intercepting ungrounded releases at the output mediation boundary.

\paragraph{Output grounding and hallucination mitigation.}
Post-generation verification tools, such as MiniCheck~\cite{tang2024minicheck}, evaluate whether emitted claims are grounded in provided documents. While output grounding checks document faithfulness, our assertion mediator evaluates the broader epistemic context: verifying that supporting evidence is authoritatively admitted, independent, temporally active, and authorized for release to the specific recipient.

\paragraph{Parametric knowledge editing.}
Techniques such as ROME~\cite{meng2022rome} and MEMIT~\cite{meng2023memit} directly modify parametric factual associations in model weights. Parametric editing operates on a different layer than persistent agent state: at a fixed accepted head, our framework keeps the explicit autobiographical claim set invariant under model substitution and denies generated content evidential standing by default.

\section{Conclusion}
\label{sec:conclusion}

Persistent AI agents require an enforceable distinction between material that is stored and claims
that are supported for release. Neither a foundation model's fluency nor a memory object's
reachability gives an individual agent evidence, authority, and valid scope to claim content
autobiographically. We formalized \emph{autobiographical assertion boundedness} through a
system-level contract spanning a typed provenance graph, source independence, controlled role
promotion, typed resolution with orthogonal flags, and generate--verify--revise assertion mediation.

Under its stated extraction, predicate-correctness, resolution, declassification, and channel
assumptions, the conditional contract result restricts release to units that pass the content or
status checks. In 24 deterministic, hand-authored conformance cases, typed mediation passed none of
19 unsafe opportunities unqualified and passed all five supported controls. The suite
validates the encoded obligations but does not measure natural-language extraction or deployed
retrieval and generation systems. By keeping reachability and support separate, the resulting
contract provides a precise basis for evaluating how persistent agents bind outward
autobiographical claims to accepted state.

\section*{Acknowledgments}
OpenAI Codex was used for language editing, formal-consistency review, and execution of
reproducibility checks. The authors reviewed the resulting manuscript, proofs, citations, and
experimental outputs and take responsibility for the final content.

\begingroup
\fontsize{7}{7.8}\selectfont
\let\oldthebibliography\thebibliography
\renewcommand\thebibliography[1]{%
  \oldthebibliography{#1}%
  \setlength{\itemsep}{0pt plus 0.2pt}%
  \setlength{\parsep}{0pt}%
}
\bibliographystyle{unsrt}
\bibliography{references}
\endgroup

\clearpage
\appendix
\begin{table*}[t]
\centering
\caption{Taxonomy of claim origins, verifiable provenance guarantees, and boundaries.}
\label{tab:claim-origins-full}
\footnotesize
\renewcommand{\arraystretch}{0.88}
\begin{tabularx}{\textwidth}{L{3.0cm} L{5.2cm} X}
\toprule
\textbf{Origin Class} & \textbf{Verifiable Provenance Guarantee} & \textbf{Epistemic Boundary} \\
\midrule
Direct capture & Hardware device, byte commitment, capture path, and timestamp & Sensor calibration and environmental fidelity \\
Attributed report & Channel attribution, cryptographic signature, and asserted scope & Informant honesty, accuracy, and temporal stability \\
Imported record & Document version hash, repository custody, and import pipeline & Historical applicability to current agent context \\
Model prior & Model architecture identifier and checkpoint metadata & Lacks evidentiary grounding in agent interaction history \\
Retrieved external text & Search query, retrieved snippet hash, source URL, and timestamp & External source reliability and future validity \\
Deterministic derivation & Exact dependency inputs, code version hash, and deterministic output & Semantic soundness of transformation logic \\
Probabilistic interpretation & Model evaluator version, input embeddings, and sealed inference & Evaluator drift and non-deterministic reproducibility \\
Hypothetical simulation & Simulation framework, parameter config, and scenario label & Explicitly non-actualized synthetic execution \\
\bottomrule
\end{tabularx}

\vspace{1.5mm}

\caption{Policy-relative role promotions and target admission constraints.}
\label{tab:promotion-examples-full}
\footnotesize
\renewcommand{\arraystretch}{0.88}
\begin{tabularx}{\textwidth}{L{3.0cm} L{3.0cm} X}
\toprule
\textbf{Source Role ($r_s$)} & \textbf{Target Role ($r_t$)} & \textbf{Governing Policy Constraint} \\
\midrule
Attributed report & Memory interpretation & Binds original report hash, interpretation evaluator version, and scope; preserves attribution \\
Memory interpretation & Recorded belief & Requires explicit stance endorsement authority and retention of underlying interpretation dependencies \\
Model prior & Marked hypothesis & Requires labeling parametric model origin; denies unqualified autobiographical standing \\
Authenticated user message & Attributed report & Creates a report object bound to sender signature, channel custody, message hash, and timestamp \\
Direct sensor capture & Observed event & Creates an observation object bound to capture device pipeline, timestamp, and integrity checks \\
Model prior / unverified retrieval & Observed event or attributed report & \textbf{Strictly prohibited}: unverified or parametric origins cannot serve as observation or report provenance \\
\bottomrule
\end{tabularx}
\end{table*}

\section{Supporting Formal Proofs}
\label{sec:appendix-proofs}

This appendix supplies proof details for Proposition~\ref{prop:promotion-traceability}
(Epistemic Provenance Traceability) and Proposition~\ref{prop:assertion-boundedness}
(Soundness of Assertion Mediation)
under their stated assumptions.

\subsection{Proof of Proposition~\ref{prop:promotion-traceability} (Epistemic Provenance Traceability)}
\label{sec:proof-traceability}

\begin{proof}
Let $h$ satisfy the assumptions of Proposition~\ref{prop:promotion-traceability}. Because
$\mathcal G_h=(\mathcal V_h,\mathcal E_h,\lambda)$ is finite and acyclic, it has a topological
ordering. We induct over that ordering to establish the stronger invariant for every admitted
vertex: it has a finite dependency trace whose terminal entries are policy-recognized source
artifacts or unavailable markers, with exact admitted versions and transitions inspectable at $h$.

\emph{Base case.} Let $v$ be a root. The root-admission assumption gives
$\operatorname{RootOK}_\kappa(v,h)$. If $\tau(v)=\textsf{ObservedEvent}$ or
$\tau(v)=\textsf{AttributedReport}$, Equation~\ref{eq:target-origin-obligations} additionally requires
$\operatorname{CaptureOriginOK}_\kappa(v,h)$ or $\operatorname{ReportOriginOK}_\kappa(v,h)$,
respectively. A root in another role retains its policy-recognized origin marker without thereby
joining $\mathcal X$. The singleton trace at $v$ satisfies the claim.

\emph{Induction step.} Let non-root $v$ be admitted by a recorded judgment
$\mathcal G_h\vdash_\kappa v_s\xrightarrow[D]{\mathrm{prom}}v:\tau(v)$. Its dependencies precede
$v$ in the topological order. By E2, each reference in $D$ resolves to its exact admitted version or
an explicit unavailable marker. For each exact version $u\in D\cap\mathcal V_h$, the induction hypothesis supplies
a finite dependency trace; attaching $v$ through the recorded edge $(u,v)$ extends that trace. An
unavailable marker in $D\cap\mathcal M_h$ is a terminal entry bound to its affected reference in the transition record.
The graph assumption preserves finiteness and acyclicity. Equation~\ref{eq:no-silent-promotion}
separately prevents a source in $\mathcal U$ from being promoted into $\mathcal X$.
\end{proof}

\subsection{Proof of Proposition~\ref{prop:assertion-boundedness} (Soundness of Assertion Mediation)}
\label{sec:proof-boundedness}

\begin{proof}
Let $z$ be a candidate response on a governed channel under context
$c=\langle h_p,h_e,p,u,s,\kappa,t_e\rangle$. We show that any output $z_{\mathrm{rel}}$ reached by
the release transition satisfies $\operatorname{Bounded}_{\mathcal L_{bio}}(z_{\mathrm{rel}},c)$.

\emph{Interception and semantic decomposition.} By E4, every candidate emission on a governed
channel is intercepted. Under the complete-extraction and correct-typing assumptions, the candidate
is decomposed into the ordered occurrences
$\operatorname{Units}_{\mathcal L_{bio}}(z)=\langle v_1,\ldots,v_n\rangle$. Their unique identifiers
occur exactly once, and the occurrences partition into content and status units. Each $v$ is bound
to a structured query $Q_v\in\mathcal L_{bio}$.

For each unit $v$, the mediator computes
$\mathcal G_{h_p};\mathcal B_{h_p}(p,u,s)\vdash_\kappa Q_v\Downarrow
d_v=\langle E,F,S,W\rangle$:
\begin{itemize}
  \item \emph{Content assertions.} The mediator emits $\alpha$ only if every conjunct of
  $\operatorname{ContentReleaseOK}(\alpha,d_\alpha,c)$ holds.
  Equation~\ref{eq:resolution-consistency} supplies nonempty, policy-valid support for a
  $\textsf{Supported}$ result. The assumed-correct $\operatorname{ExactMatch}$,
  $\operatorname{QualOK}$, $\operatorname{DisclosureOK}$, and $\operatorname{HeadUseOK}_\kappa$
  implementations establish the remaining conjuncts.
  \item \emph{Status responses.} The renderer receives $\mathsf{View}_c(d_\beta)$, never raw
  $W(d_\beta)$. A status unit is released only if
  $\beta\in\operatorname{RenderSet}_\kappa(\mathsf{View}_c(d_\beta),c)$ and it passes
  $\operatorname{HeadUseOK}$. Equation~\ref{eq:render-noninterference} gives equal render sets for
  equal views; where policy requires internal outcomes to be indistinguishable,
  Equation~\ref{eq:view-collapse} maps them to the same view.
\end{itemize}

If any semantic unit fails verification, that candidate is blocked. A revision re-enters extraction,
and E4 prevents bypass. Hence the release transition is reachable only when every covered unit
satisfies its corresponding release relation. By Definition~\ref{def:assertion-boundedness}, the
emitted $z_{\mathrm{rel}}$ is bounded.
\end{proof}

\begin{table*}[t]
\centering
\caption{Adversarial evaluation tracks and target epistemic failure modes.}
\label{tab:epistemic-eval-plan-full}
\footnotesize
\renewcommand{\arraystretch}{0.88}
\begin{tabularx}{\textwidth}{L{2.6cm} L{5.2cm} X}
\toprule
\textbf{Evaluation Track} & \textbf{Target Scenarios} & \textbf{Primary Failure Mode Exposed} \\
\midrule
Origin \& promotion & Model priors, unverified web text, simulations, synthetic summaries & Unauthorized elevation of ungrounded inputs into autobiographical history \\
Source independence & Surface paraphrases, multi-URL duplicates sharing upstream roots & False corroboration and synthetic consensus through redundant retrieval \\
Scope \& disclosure & Recipient-restricted data, task permissions, confidential records & Unauthorized cross-context data leakage and privacy violations \\
Conflict \& time & Contradictory records, expired permissions, retracted evidence & Stale assertion release, unhandled disagreement, and invalidation gaps \\
Memory poisoning & Indirect prompt injections in ingested docs, snippets, and tools & Persistent compromise of agent profile and control flow across sessions \\
Substrate upgrade & Model replacement ($m_1 \to m_2$) under fixed accepted head $h$ & Hallucinatory embellishment or parametric drift masquerading as biography \\
Assertion extraction & Presuppositions, negations, modality shifts, compound clauses & Evading assertion mediation via indirect or non-declarative framing \\
\bottomrule
\end{tabularx}
\end{table*}

\begin{table*}[t]
\centering
\caption{Architectural comparison of objects, enforcement surfaces, and provenance contracts.}
\label{tab:novelty-matrix-full}
\scriptsize
\renewcommand{\arraystretch}{0.88}
\setlength{\tabcolsep}{3.5pt}
\begin{tabularx}{\textwidth}{L{2.7cm} L{2.5cm} L{2.9cm} L{3.3cm} X}
\toprule
\textbf{Framework} & \textbf{Primary Object} & \textbf{Enforcement Surface} & \textbf{Provenance Contract} & \textbf{Model-Replacement Invariant} \\
\midrule
Continuity Kernel~\cite{he2026continuitykernel} & State-transition proposal & Accepted-head activation & Proposal / parent hash chain; epistemic support outside contract & State protocol remains invariant across backbones \\
MemIR~\cite{jin2026memir} & Typed memory atoms & Memory write and projection & Evidence / cue / claim roles; source-root graph not modeled & Not evaluated \\
MemTX~\cite{li2026memtx} & Belief record \& action & Commit gate \& action gate & Record dependency lineage; source-root independence not enforced & Multi-backbone support; no fixed-head invariant \\
ConsistencyGate~\cite{zhang2026consistencygate} & Candidate memory fact & Write-time filter & Self-consistency score; does not track upstream source roots & Multi-backbone support; no fixed-head invariant \\
MiniCheck~\cite{tang2024minicheck} & Output claim atom & Post-generation gate & Document grounding check; source-root graph not modeled & Not evaluated \\
This paper & Claim vertex \& response unit & Typed resolution \& release gate & Exact versioned dependency DAG with policy-relative source roots & Fixed head determines the accepted autobiographical claim set \\
\bottomrule
\end{tabularx}

\vspace{1.8mm}

\caption{Comparison of temporal validity, correction models, disclosure scoping, and release contracts.}
\label{tab:novelty-matrix-release-full}
\scriptsize
\renewcommand{\arraystretch}{0.88}
\setlength{\tabcolsep}{3.5pt}
\begin{tabularx}{\textwidth}{L{2.1cm} L{3.1cm} L{2.9cm} L{3.1cm} X}
\toprule
\textbf{Framework} & \textbf{Temporal / Correction Model} & \textbf{Disclosure Scoping} & \textbf{Assertion Release Contract} & \textbf{Withholding Privacy Safety} \\
\midrule
Continuity Kernel & Fresh parent head; forward state succession & Outside stated contract & Outside stated contract & Outside stated contract \\
MemIR & Temporal grounding; correction not modeled & Retrieval / provenance scope & Normalized fact projection; non-bypass not specified & Not specified \\
MemTX & Record validity \& typed cascade repair & Action-focused permissions & Outward assertion release outside contract & Not specified \\
ConsistencyGate & Temporal validity not modeled & Outside stated contract & Outside stated contract & Not specified \\
MiniCheck & Temporal validity not modeled & Outside stated contract & Claim-level grounding evaluator & Not specified \\
This paper & Validity intervals, transitive invalidation, and emission-head checks & Granular principal, purpose, and relationship scope & Decoupled typed content and status release checks & Recipient-safe uniform status rendering \\
\bottomrule
\end{tabularx}
\end{table*}

\section{Taxonomies and Role Promotion Rules}
\label{sec:appendix-schemas}
\label{sec:appendix-origin-taxonomy}
\label{sec:appendix-promotion-rules}

Table~\ref{tab:claim-origins-full} details the origin taxonomy and provenance guarantees, and
Table~\ref{tab:promotion-examples-full} illustrates instances of the policy-relative judgment
$\mathcal G_h\vdash_\kappa v_s\xrightarrow[D]{\mathrm{prom}}v_t:r_t$.

\section{Proposed End-to-End Evaluation Protocol and Tracks}
\label{sec:appendix-eval-protocol}

Table~\ref{tab:epistemic-eval-plan-full} defines a proposed seven-track adversarial evaluation
design for long-term agent deployments. The present paper does not report this end-to-end study.

\paragraph{Baselines and Ablation Configurations.}
\label{sec:appendix-baselines-ablations}
The proposed study would evaluate five deployment baselines: (1)~a parametric agent using
zero-shot system prompts without episodic memory; (2)~a flat memory store using dense semantic
search over past transcripts without provenance parsing; (3)~standard RAG injecting retrieved
chunks with grounding instructions; (4)~provenance-tagged retrieval without root-DAG analysis or
invalidation checks; and (5)~a state-bounded baseline that filters claim objects at prompt assembly
without post-generation mediation. Retrieval depth, chunking, model, and decoding settings must be
fixed and reported when the study is executed. Proposed ablations isolate source-root independence
($\operatorname{Roots}_\kappa$), transitive dependency invalidation
(Equation~\ref{eq:correction-propagation}), emission-head timing ($\mathsf{HeadUseOK}_\kappa$),
and recipient-safe declassification ($\mathsf{View}_c$).

\paragraph{Falsification Criteria.}
\label{sec:appendix-falsification}
The architecture would be empirically challenged if: (1)~semantic extraction misses covered
assertions above a deployment's declared safety tolerance; (2)~mediation reduces valid-query
coverage beyond its declared utility bound; or (3)~a simpler rule achieves equivalent safety and
coverage at lower latency. These thresholds should be selected before executing the end-to-end
study; the conformance suite does not estimate them.

\section{Conformance Suite Fixture and Dependency Topology}
\label{sec:appendix-traces}

The deterministic conformance suite evaluates the typed resolver and assertion mediator across 24
cases in six threat tracks. Corroboration cases require two roots that satisfy
$\operatorname{Indep}_\kappa(\cdot,\cdot;h)$.

\paragraph{Dependency Topology and Execution.}
The evaluation fixture encodes acyclic dependency topologies through accepted claims, evidence
items, and terminal-root identifiers; the reference resolver consumes flattened root references
rather than traversing a general $\mathcal G_h$. The cases cover model priors,
prompt injection, simulated or hallucinated origins, overlapping root topologies, scope and erasure,
temporal bounds, conflict, invalidation, supersession, and emission-head races, alongside five
supported controls. Table~\ref{tab:epistemic-trace-oracle}
(Section~\ref{sec:epistemic-threats-evaluation}) specifies the per-case oracle, fixture roots, and
mediation dispositions.

\section{Sociotechnical Governance, Privacy, and Erasure}
\label{sec:appendix-governance}

\subsection{Consent Revocation and Existence Leakage}
\label{sec:appendix-consent-leakage}

\begin{proposition}[Existence Side Channel under Distinguishable Rendering]
Let $\mathcal O(Q,c)$ be a deterministic output class observable to an adversary who knows the
rendering policy and can distinguish the classes used for $\mathsf{Unknown}$ and
$\mathsf{Withheld}$. If those internal branches always map to distinct output classes, observing
$\mathcal O(Q,c)$ reveals which branch occurred for that query.
\end{proposition}

When policy requires those branches to remain indistinguishable, Equation~\ref{eq:view-collapse}
maps the corresponding decisions to $\bot_{unk}$. Equation~\ref{eq:render-noninterference} then gives
that common view one render set, closing this branch-distinguishing channel under the proposition's
observation assumptions.

\subsection{Relational Privacy and Multi-Party Contexts}
\label{sec:appendix-relational-privacy}

A deployment may enforce three relational privacy rules: (1)~\textbf{Isolation:} claim objects
partition private statements from shared context; (2)~\textbf{Relational Consent:} disclosing
$p_j$'s statements to $p_i$ requires authorization from $p_j$; and (3)~\textbf{Asymmetric
Revocation:} admitted consent changes update later projections without rewriting historical
state.

\subsection{Erasure-Aware Lineage Design}
\label{sec:appendix-gdpr-erasure}

One possible deployment pattern separates evidence retention from lineage metadata: encrypted
evidence may be made inaccessible through key destruction; affected nodes then transition to
\textsf{Unavailable}; derived claims are recomputed under
Equation~\ref{eq:correction-propagation}; and non-sensitive lineage metadata may be retained where
policy permits. Whether this pattern satisfies GDPR Article 17 or another erasure obligation
depends on the jurisdiction, controller obligations, data model, and implementation. The present
artifact neither implements cryptographic shredding nor establishes legal compliance.

\section{Related Work Comparison}
\label{sec:appendix-related-work}

Tables~\ref{tab:novelty-matrix-full} and~\ref{tab:novelty-matrix-release-full} summarize the
comparison with selected related systems.

\end{document}